%% file: samplepaper.tex
\documentclass[runningheads,orivec]{llncs}
\usepackage[T1]{fontenc}
\usepackage{graphicx}
\usepackage{hyperref}
\usepackage{color}
\input{config}

\begin{document}
\title{String Matching with a Dynamic Pattern}
\author{Bruno Monteiro\orcidID{0009-0009-9435-5323} \and
Vinicius dos Santos\orcidID{0000-0002-4608-4559}}
\authorrunning{B. Monteiro \and V. dos Santos}
\institute{Computer Science Department \\
Federal University of Minas Gerais, Brazil \\
\email{malettabruno@gmail.com, viniciussantos@dcc.ufmg.br}}
\maketitle              % typeset the header of the contribution
\begin{abstract}
In this work, we tackle a natural variation of the String Matching Problem on the case of a dynamic pattern, that is, given a static text $T$ and a pattern $P$, we want to support character additions and deletions to the pattern, and after each operation compute how many times it occurs in the text. We show a simple and practical algorithm using Suffix Arrays that achieves $\mathcal O(\log |T|)$ update time, after $\mathcal O(|T|)$ preprocess time. We show how to extend our solution to support substring deletion, transposition (moving a substring to another position of the pattern), and copy (copying a substring and pasting it in a specific position), in the same time complexities. Our solution can also be extended to support an online text (adding characters to one end of the text), maintaining the same amortized bounds.

\keywords{Strings \and Algorithms \and Suffix array \and String matching}
\end{abstract}

\input{text/intro}
\input{text/sa_sr}
\input{text/concatenation}
\input{text/partition}

\input{text/improving}
\input{text/substring}
\input{text/online}
\input{text/conclusion}

\newpage

\bibliographystyle{splncs04}
\bibliography{references}

\newpage
\section*{Appendix: computation of the suffix range}
\input{text/computation}

\end{document}

%% file: config.tex
\usepackage{mdframed}
\usepackage[table]{xcolor}

\usepackage[cfg=config/listofsymbols.cls]{nomencl}

\usepackage{amssymb}
\makenomenclature
\usepackage{etoolbox}

\usepackage{algorithm}
\usepackage[noend]{algpseudocode}
\patchcmd{\listof}{1.5em}{0pt}{}{}

\usepackage[noredefwarn,toc]{glossaries}
\usepackage{calc}

\usepackage{mdframed}

\usepackage{makecell}

\newcommand{\SA}{\mathcal{SA}}
\newcommand{\ISA}{\mathcal{ISA}}
\newcommand{\LCP}{\mathcal{LCP}}
\newcommand{\OO}{\mathcal{O}}
\newcommand{\SR}{\mathcal{SR}}

\usepackage{tikz}
\usetikzlibrary{automata,arrows,tikzmark,calc,positioning}
\tikzset{->,>=stealth',auto,node distance=2cm,thick,initial text=} 
\tikzset{
    position/.style args={#1:#2 from #3}{
        at=(#3.#1), anchor=#1+180, shift=(#1:#2)
    }
}

\usepackage{pgfplots}
\pgfplotsset{width=10cm,compat=1.9}

\renewcommand{\tt}{\ttfamily}

\makeatletter
\newcommand{\xRightarrow}[2][]{\ext@arrow 0359\Rightarrowfill@{#1}{#2}}
\makeatother

\usepackage{xpatch}
\makeatletter
\xpretocmd{\@chapter}{\addtocontents{loa}{\protect\addvspace{10pt}}}{}{}{}%
\makeatother

%% file: text/intro.tex
\section{Introduction}

The String Matching Problem consists of, given two strings usually called the \textit{text} and the \textit{pattern}, computing the indices where the pattern occurs in the text. After Knuth, Morris, and Pratt~\cite{kmp,mp} settled it with a linear-time solution in 1970, work was done on variations of the problem. 
Fischer and Paterson~\cite{fischer} introduced the problem of \textit{Pattern Matching with Wildcards}, in which we can have ``don't-care'' symbols in the text or the pattern, that match with any other symbol. Another studied variation is the \textit{Approximate Pattern Matching}, which consists of finding occurrences of the pattern subject to at most $k$ mismatches~\cite{ivanov,landau}.
Due to its wide applicability, pattern matching and related problems remain an active and well-studied research topic~\cite{das2022internal,kociumaka2024internal}.

Another problem of interest is the Indexing Problem. This problem can be solved \textit{offline} (if all patterns are known in advance) in linear time \cite{ahocorasick}, for a constant size alphabet. In the \textit{online} scenario, we want to preprocess only the text and answer queries efficiently. Classical solutions for this problem use suffix trees or suffix arrays.
Both data structures can be computed directly in linear time and space \cite{farach,karkkainen}, assuming the alphabet symbols are integers bounded by $|T|^c$, for text $T$ and $c \in \OO(1)$. Of course, if the alphabet is unbounded, in the comparison model there is a lower bound of $\Omega(|T| \log |T|)$ (from sorting) for building these data structures \cite{cormen}. When it comes to query time, the data structures achieve different bounds (see \autoref{tab:indexing}).

\begin{table}[ht]
    \centering
    \caption{Deterministic worst-case query time for the Indexing Problem (for pattern $P$ and text $T$) across different data structures that can be built in linear time and space.}
    \begin{tabular}{|c|c|c|}
        \hline
        \cellcolor{gray!20} Data Structure & \cellcolor{gray!20} Query Time & \cellcolor{gray!20} Source \\
        \hline\hline
        Suffix Tree & \makecell{$\OO(|P| \log |\Sigma|)$ \\ or \\ $\OO(|P| + \log |T|)$} & \makecell{\cite{weiner} \\ \\ \cite{cole2003}} \\
        \hline
        Suffix Array & $\OO(|P| + \log |T|)$ & \cite{myers} \\
        \hline
        Suffix Tray & $\OO(|P|+ \log |\Sigma|)$ & \cite{cole2006} \\
        \hline
    \end{tabular}
    \label{tab:indexing}
\end{table}

It is also possible to maintain the $\OO(|P| + \log |T|)$ query time while supporting updates of adding a symbol to the end of the text (also called \textit{online text}) in $\OO(\log |T|)$ worst-case time \cite{amir2005}. Studies have been made in the case of a \textit{dynamic} text, in which we can add or remove symbols in arbitrary positions of the text while supporting efficient queries \cite{alstrup2000pattern,ferragina1998,gu}.

The case of a dynamic pattern has been given less attention. Amir and Kondratovsky~\cite{amir2018} were the first to give a sub-linear solution for updating the pattern and querying for the number of occurrences of the pattern in the text, the \textit{modified pattern reporting problem}. They show how to maintain a pattern subject to symbol change (insertion and deletion) in $\OO(\log |T|)$ time, and substring copy-pasting and deletion in $\OO(\log |T| + \ell)$ time, with $\ell$ being the size of the modified substring. This is done after $\OO(|T| \sqrt{\log |T|})$ preprocessing time and memory.

Using several insights regarding the suffix array, we achieve better time bounds for the aforementioned problem (see \autoref{tab:complexities_static_intro}), with a simpler algorithm. We show how we can maintain a dynamic pattern, subject to character addition or deletion and substring deletion, transposition (moving a substring to another position of the pattern), and copy (copying a substring and pasting it in a specific position). After every update to the pattern, we can output the number of times the pattern occurs in the text.

\begin{table}[ht]
    \centering
    \caption{Overview of the time bounds we achieve for the \textit{modified pattern reporting problem}. Here, edition is an addition or deletion, and $\ell$ is the size of the modified substring.}
    \begin{tabular}{|c|c|c|}
        \hline
        \cellcolor{gray!20} Operation & \cellcolor{gray!20} \cite{amir2018} & \cellcolor{gray!20} This work \\
        \hline\hline
        Preprocess Time and Space & $\OO(|T| \sqrt{\log |T|})$ & $\OO(|T|)$ \\
        \hline
        Pattern Symbol Edition & $\OO(\log |T|)$ & $\OO(\log |T|)$ \\
        \hline
        Pattern Substring Edition & $\OO(\log |T| + \ell)$ & $\OO(\log |T|)$ \\
        \hline
    \end{tabular}
    \label{tab:complexities_static_intro}
\end{table}

In addition, we apply our algorithm for an online text, that is, supporting addition of a symbol to one end of the text. We achieve amortized logarithmic bounds (see \autoref{tab:complexities_online_intro}). 
While Amir and Kondratovsky~\cite{amir2018} outline a suffix tree based approach claiming logarithmic worst-case bounds, this work provides a simpler suffix array based solution with clear logarithmic amortized bounds for all operations.
%The paper by \citef{amir2018} claims worst-case logarithmic bounds, but we could not understand if this is really the case, and we could not get in touch with the authors, despite our best efforts.

\begin{table}[ht]
    \centering
    \caption{Overview of the time bounds we achieve for the \textit{modified pattern reporting problem} with an online text. Here, edition is an addition or deletion, and $\ell$ is the size of the modified substring.}
    \begin{tabular}{|c|c|c|}
        \hline
        \cellcolor{gray!20} Operation & \cellcolor{gray!20} \cite{amir2018} & \cellcolor{gray!20} This work \\
        \hline\hline
        Preprocess Time and Space & $\OO(|T| \sqrt{\log |T|})$ & $\OO(|T|)$ \\
        \hline
        Pattern Symbol Edition & $\OO(\log |T|)$ & $\OO(\log |T|)^\star$ \\
        \hline
        Pattern Substring Edition & $\OO(\log |T| + \ell)$ & $\OO(\log |T|)^\star$ \\
        \hline
        Text Symbol Extension & $\OO(\log |T|)$ & $\OO(\log |T|)^\star$ \\
        \hline
    \end{tabular}
    \label{tab:complexities_online_intro}
\end{table}

For convenience of presentation, we define \autoref{prob:dyn_pattern_static_text}, which can be shown to be equivalent to the \textit{modified pattern reporting problem}.

\begin{problem}[Dynamic Pattern and Static Text Matching]
    \label{prob:dyn_pattern_static_text}
    
    \parindent0pt
    \textbf{Input}: A string $T$ that represents the text, several updates to the pattern string $P$ that can be one of the following.

    \begin{enumerate}
        \item Pattern search: set the current pattern to a new pattern provided as input;
        \item Pattern symbol edition: addition or deletion of some character of the current pattern;
        \item Pattern substring edition: substring deletion, transposition (moving the substring to another position), or copy on the current pattern.
    \end{enumerate}

    \textbf{Output}: After every operation, the number of times the current pattern occurs in the text.
\end{problem}

The solution from the literature~\cite{amir2018} uses suffix trees. Our solution with suffix arrays is considerably simpler, reduces the preprocessing time and space requirement, and improves upon the time required for the substring operations.

Note that in this work we do not assume that the text $T$ is necessarily much larger than the pattern $P$ ($|T| \gg |P|$). Even though we assume $|T| > |P|$, factors of, for example, $\log |T|$ and $|P|$ are both taken into account in the time complexity analysis. That is, we are interested in the cases when the pattern can be large relative to the text.
% If it is always small, we can solve \autoref{prob:dyn_pattern_static_text} by using text indexing.

%% file: text/sa_sr.tex
\section{Suffix Arrays and Suffix Ranges}

The suffix array was first developed by Manber and Myers~\cite{myers} as a more space-efficient alternative to Suffix Trees, a data structure developed for solving the Indexing Problem \cite{myers,weiner}.

\begin{definition}[Suffix Array of a String]
    The \textit{suffix array} of a string $S$, denoted as $\SA(S)$, is an array of length $|S|$ such that $\SA(S)[i]$ stores the starting index of the $i$-th smallest suffix of $S$, in lexicographical order.
    \label{def:sa}
\end{definition}

Note that there are no ties, so the suffix array is unique. It is very common to use, along with the suffix array, the \textit{lcp} array.

\begin{definition}[Longest Common Prefix Array of a String]
    The \textit{longest common prefix array} (or simply $lcp$ array) of a string $S$, denoted as $\LCP(S)$, is an array of size $|S|-1$ such that $\LCP(S)[i] = lcp(S_{\SA[i]}, S_{\SA[i+1]})$, that is, $\LCP(S)[i]$ is the $lcp$ between the $i$-th and $(i+1)$-th smallest suffixes of $S$.
    \label{def:lcp}
\end{definition}

% See \autoref{fig:sa_example2} for an example.
Both the suffix array and the $lcp$ array can be computed in linear time \cite{karkkainen,kasai,kim}, assuming the alphabet symbols are integers bounded by $|T|^c$, for text $T$ and $c \in \OO(1)$.

% \begin{figure}[ht]
%     \centering
%     \begin{tabular}{cccl}
%         $i$ & \tiny $\SA[i]$ & \tiny $\LCP[i]$ & \tiny $S_{\SA[i]}$ \\
%         \hline
%         0  & 10 & 1 & \tt i \\
%         1  & 7  & 1 & \tt ippi \\
%         2  & 4  & 4 & \tt issippi \\
%         3  & 1  & 0 & \tt ississippi \\
%         4  & 0  & 0 & \tt mississippi \\
%         5  & 9  & 1 & \tt pi \\
%         6  & 8  & 0 & \tt ppi \\
%         7  & 6  & 2 & \tt sippi \\
%         8  & 3  & 1 & \tt sissippi \\
%         9  & 5  & 3 & \tt ssippi \\
%         10 & 2  &   & \tt ssissippi \\
%     \end{tabular}
%     \caption{$\SA$ and $\LCP$ of $S = \text{\tt mississippi}$.}
%     \label{fig:sa_example2}
% \end{figure}

\begin{lemma}
Let $S$ be a string and $\SA$ its suffix array. For $0 \leq i < j < |S|$, if there is a common prefix between $S_{\SA[i]}$ and $S_{\SA[j]}$ of length $\ell$, then $lcp(S_{\SA[i]}, S_{\SA[k]}) \geq \ell$, for all $i < k < j$.
\label{lem:lcprange}
\end{lemma}

\begin{proof}
From definition of suffix array, we know that $S_{\SA[i]} \prec S_{\SA[k]} \prec S_{\SA[j]}$. By contradiction, assume that the statement is false. This means that, for some $i < k < j$, there is an integer $x < \ell$ such that $lcp(S_{\SA[i]}, S_{\SA[k]}) = x$ and $S_{\SA[i]}[x] < S_{\SA[k]}[x]$. But since $lcp(S_{\SA[i]}, S_{\SA[j]}) \geq \ell$, we have that $S_{\SA[i]}[x] = S_{\SA[j]}[x]$, implying $S_{\SA[j]} \prec S_{\SA[k]}$, a contradiction.
\end{proof}

\begin{lemma}
    The $lcp$ between any two suffixes of $S$, $S_i$ and $S_j$, is equal to the minimum over the range in the $lcp$ array (assuming $\ISA[i] < \ISA[j]$):
    
    $$ lcp(S_i, S_j) = \min_{\ISA[i] \leq k < \ISA[j]} \LCP[k], $$

    This also holds for any set of strings sorted lexicographically.
\label{lem:lcp_query}
\end{lemma}

\begin{proof}
    Let $\ell = lcp(S_i, S_j)$. From \autoref{lem:lcprange}, we know that the first $\ell$ characters of $S_i$ are equal to those of $S_{\SA[\ISA[i]+1]}$ (the next suffix in the suffix array order), and are also equal to those of $S_{\SA[\ISA[i]+2]}$, and so on until $S_{j}$. Therefore, the $\LCP$ values of these positions are at least $\ell$. And they can not be all greater than $\ell$, otherwise we would have $lcp(S_i, S_j) > \ell$. So the minimum of the $\LCP$ values of the range is exactly $\ell$. All these arguments are also valid for any set of strings sorted lexicographically.
\end{proof}

From \autoref{lem:lcp_query} we can solve $lcp$ queries of arbitrary suffixes in constant time, since range minimum queries can be solved in constant time with linear time construction \cite{bender}.

\subsection{Suffix range}

The reason why suffix arrays are useful for the Indexing Problem is because the occurrence of any pattern defines a \textit{range} in the suffix array, which we call \textit{suffix range} (see \autoref{def:suffixrange}).

\begin{definition}[Suffix Range]
    Let $P, T$ be strings, and let $S$ be the set of indices where $P$ occurs in $T$. The \textit{suffix range} of $P$ with respect to $T$, denoted by $\SR(P, T)$ is the set $\{i : \SA(T)[i] \in S\}$. From \autoref{lem:suffixrange}, this set is a range of indices, so we can represent it by a range $[\ell, r)$, meaning that the suffix range is $\ell, \ell+1, \dots, r-1$. If $P$ does not occur in $T$, we represent the suffix range by the empty range $[0, 0)$.
    \label{def:suffixrange}
\end{definition}

\begin{lemma}
    If $P$ occurs in $T$, then $\SR(P, T)$ is a range of indices.
\label{lem:suffixrange}
\end{lemma}

\begin{proof}
    Let $\SA$ be the suffix array of $T$ and let $S$ be the set of indices where $P$ occurs in $T$. By contradiction, assume that there are indices $i < k < j$ such that $\SA[i], \SA[j] \in S$ and $\SA[k] \notin S$. We know that $lcp(T_{\SA[i]}, T_{\SA[j]}) \geq |P|$, so from \autoref{lem:lcprange} we have that $lcp(T_{\SA[i]}, T_{\SA[k]}) \geq |P|$, so $\SA[k]$ is an occurrence of $P$, a contradiction.
\end{proof}

It turns out that $\SR(P, T)$ can be computed in $\OO(|P| + \log |T|)$ using binary search and some clever $lcp$ insights \cite{myers}.

%% file: text/concatenation.tex
\section{Suffix range concatenation}

Our main insight is the fact that if we have $|A|$, $|B|$, $\SR(A, T)$, and $\SR(B, T)$, we can compute $\SR(A \circ B, T)$ efficiently, as we will see in \autoref{th:concat} and \autoref{alg:concat}. Here, $A \circ B$ denotes the concatenation of strings $A$ and $B$.

\begin{lemma}
    If $A \circ B$ occurs in $T$, then $\SR(A \circ B, T)$ is a sub-range of $\SR(A, T)$.
\label{lem:concat}
\end{lemma}
\begin{proof}
    All of the suffixes from $\SR(A \circ B, T)$ have $A$ as a prefix, therefore they are also in $\SR(A, T)$.
\end{proof}

\begin{theorem}
    Given $|A|$, $|B|$, $\SR(A, T)$, and $\SR(B, T)$, we can compute $\SR(A \circ B, T)$ in $\OO(\log |T|)$ time.
    \label{th:concat}
\end{theorem}

\begin{proof}
    From \autoref{lem:concat}, we want to find a sub-range of $\SR(A, T)$, so it suffices to find its first and last positions. Let $[\ell, r)$ be the range $\SR(A, T)$. We want to find the first $i \in [\ell, r)$ such that $T_{\SA[i] + |A|}$ has $B$ as a prefix. But note that, since $T_{\SA[\ell]}, T_{\SA[\ell+1]}, \dots, T_{\SA[r-1]}$ all have an $lcp$ of at least $|A|$, then $T_{\SA[\ell] + |A|}, T_{\SA[\ell+1] + |A|}, \dots, T_{\SA[r-1] + |A|}$ appear in this order in the suffix array, because their lexicographical comparison is not decided by their first $|A|$ characters, so skipping them maintain their order. Therefore, we can do a binary search on the range $[\ell, r)$, and when checking some suffix, we look at the position of the suffix that skips $|A|$ characters from that suffix, and check if it is in $\SR(B, T)$, using the $\ISA$. We can therefore find the first and last position in $[\ell, r)$ that correspond to $\SR(A \circ B, T)$, with two binary searches.
\end{proof}

A visual representation of \autoref{lem:concat} is shown in \autoref{fig:concat_example}. The algorithm is implemented in \autoref{alg:concat}, using the type of binary search that maintains that the answer is in the half-open range $[\ell, r)$.

\begin{figure}[ht]
    \centering
    \begin{tabular}{ccccl}
        & $i$ & \tiny $\SA[i]$ & \tiny $\LCP[i]$ & \tiny $S_{\SA[i]}$ \\
        \cline{2-5}
                          & 0                   & 13 & 1 & \tt a \\
                          & 1                   & 10 & 1 & \tt aaca \\
        \tikzmark{black1} & \tikzmark{black2}2  & 8  & 3 & \tt {\color{red}aba}aca \\
        \tikzmark{blue1}  & \tikzmark{blue2}3   & 6  & 5 & \tt {\color{red}aba}{\color{blue}ba}aca \\
        \tikzmark{blue3}  & \tikzmark{blue4}4   & 4  & 3 & \tt {\color{red}aba}{\color{blue}ba}baaca \\
        \tikzmark{black3} & \tikzmark{black4}5  & 0  & 1 & \tt {\color{red}aba}cabababaaca \\
        \tikzmark{black5} & \tikzmark{black6}6  & 11 & 3 & \tt aca \\
                          & 7                   & 2  & 0 & \tt acabababaaca \\
        \tikzmark{blue5}  & \tikzmark{blue6}8   & 9  & 2 & \tt {\color{blue}ba}aca \\
        \tikzmark{blue7}  & \tikzmark{blue8}9   & 7  & 4 & \tt {\color{blue}ba}baaca \\
                          & 10                  & 5  & 2 & \tt {\color{blue}ba}babaaca \\
                          & 11                  & 1  & 0 & \tt {\color{blue}ba}cabababaaca \\
                          & 12                  & 12 & 2 & \tt ca \\
        \tikzmark{black7} & \tikzmark{black8}13 & 3  &   & \tt cabababaaca \\
    \end{tabular}
    \begin{tikzpicture}[overlay,remember picture]
        \draw[-]  ($(pic cs:black1)+(-8pt, 4pt)$) -- ($(pic cs:black2)+(-5pt, 4pt)$);
        \draw[-]  ($(pic cs:black1)+(-8pt, 4pt)$) -- ($(pic cs:black5)+(-8pt, 4pt)$);
        \draw[->] ($(pic cs:black5)+(-8pt, 4pt)$) -- ($(pic cs:black6)+(-5pt, 4pt)$);
        
        \draw[thick,blue,-]  ($(pic cs:blue1)+(-16pt, 4pt)$) -- ($(pic cs:blue2)+(-5pt,  4pt)$);
        \draw[thick,blue,-]  ($(pic cs:blue1)+(-16pt, 4pt)$) -- ($(pic cs:blue5)+(-16pt, 4pt)$);
        \draw[thick,blue,->] ($(pic cs:blue5)+(-16pt, 4pt)$) -- ($(pic cs:blue6)+(-5pt,  4pt)$);
        
        \draw[thick,blue,-]  ($(pic cs:blue3)+(-24pt, 4pt)$) -- ($(pic cs:blue4)+(-5pt,  4pt)$);
        \draw[thick,blue,-]  ($(pic cs:blue3)+(-24pt, 4pt)$) -- ($(pic cs:blue7)+(-24pt, 4pt)$);
        \draw[thick,blue,->] ($(pic cs:blue7)+(-24pt, 4pt)$) -- ($(pic cs:blue8)+(-5pt,  4pt)$);
        
        \draw[-]  ($(pic cs:black3)+(-32pt, 4pt)$) -- ($(pic cs:black4)+(-5pt,  4pt)$);
        \draw[-]  ($(pic cs:black3)+(-32pt, 4pt)$) -- ($(pic cs:black7)+(-32pt, 4pt)$);
        \draw[->] ($(pic cs:black7)+(-32pt, 4pt)$) -- ($(pic cs:black8)+(-5pt,  4pt)$);
    \end{tikzpicture}
    \caption{For $S = \text{\tt abacabababaaca}$, in red, $\SR(\text{\tt aba}, S)$, and, in blue, $\SR(\text{\tt ba}, S)$. The arrows illustrate the fact that the suffixes from $\SR(\text{\tt aba}, S)$ skipping $|\text{\tt aba}| = 3$ characters occur in increasing order, so we can use binary search to find the ones that end up in $\SR(\text{\tt ba}, S)$, as in \autoref{th:concat}.}
    \label{fig:concat_example}
\end{figure}

Let us see how to delete some characters from the beginning or the end of the pattern, and update the suffix range accordingly. For this, we need the following lemma.

\begin{lemma}
    Given any index $i \in \SR(P, T)$, we can find $\SR(P, T)$ in $\OO(\log |T|)$ time.
    \label{lem:extend_sr}
\end{lemma}

\begin{proof}
    We just need to find $\ell = \min_{j \leq i}\{j : lcp(T_{\SA[j]}, T_{\SA[i]}) \geq |P|\}$ and $r = \min_{i < j}\{j : lcp(T_{\SA[j]}, T_{\SA[i]}) < |P|\}$, and this gives us $\SR(P, T) = [\ell, r)$. Both these indices can be found with a binary search using $lcp$ queries, since by \autoref{lem:lcp_query} the $lcp$ value is monotonic when we increase the range: if $k < j \leq i$, $lcp(T_{\SA[k]}, T_{\SA[i]}) \leq lcp(T_{\SA[j]}, T_{\SA[i]})$, and symmetrically for the other direction.
\end{proof}

\begin{algorithm}[ht]
    \caption{Suffix Range Concatenation.}
    \hspace*{\algorithmicindent} \textbf{Input}: $|A|, |B|, \SR(A, T), \SR(B, T)$.
    \\
    \hspace*{\algorithmicindent} \textbf{Output}: $\SR(A \circ B, T)$.
    \\
    \hspace*{\algorithmicindent} \textbf{Time complexity}: $\OO(\log |T|)$.
    \bigskip
    
    \begin{algorithmic}
        \State $(lb, rb) \gets \SR(B, T)$
        \\
        
        \State $(\ell, r) \gets \SR(A, T)$ \Comment{Compute first index of the answer}
        \While{$\ell < r$}
            \State $m \gets \left\lfloor \frac{\ell+r}{2} \right\rfloor$
            \If{$\SA[m] + |A| = |T|$ or $\ISA[\SA[m] + |A|] < lb$}
                \State $\ell \gets m+1$
            \Else
                \State $r \gets m$
            \EndIf
        \EndWhile
        
        \State $first \gets \ell$
        \\
        
        \State $(\ell, r) \gets \SR(A, T)$ \Comment{Compute last index of the answer}
        \While{$\ell < r$}
            \State $m \gets \left\lfloor \frac{\ell+r}{2} \right\rfloor$
            \If{$\SA[m] + |A| = |T|$ or $\ISA[\SA[m] + |A|] < rb$}
                \State $\ell \gets m+1$
            \Else
                \State $r \gets m$
            \EndIf
        \EndWhile
        
        \State $last \gets \ell$
        \\
        
        \If{$l = r$}
            \State \Return $[0, 0)$ \Comment{Empty range: no occurrences}
        \Else
            \State \Return $[first, last)$
        \EndIf
    \end{algorithmic}
    \label{alg:concat}
\end{algorithm}

\begin{lemma}
    If $P$ occurs in $T$, given $\SR(P, T)$ we can compute the suffix range of $P$ with $k < |P|$ characters removed from the beginning or the end in $\OO(\log |T|)$ time.
    \label{lem:remove_from_pattern}
\end{lemma}

\begin{proof}
    Using \autoref{lem:extend_sr} it is easy to remove some amount of characters from the end of the pattern: from \autoref{lem:concat}, we know that the suffix range contains the suffix range we had. So we can take any index of the suffix range we had and extend it to find the new suffix range.
    
    To remove some amount of characters from the beginning, we apply a similar strategy. Let $P'$ be the pattern $P$ with the first $k$ characters removed. If $i \in \SR(P, T)$, then $\ISA(\SA[i]+k) \in \SR(P', T)$. Therefore, we can again use \autoref{lem:extend_sr} and compute the new suffix range.
\end{proof}

Now we are ready to tackle arbitrary character edition in the pattern.

\begin{theorem}
    Let $P$ be a pattern that occurs in the text $T$. Let $P'$ be $P$ with the $i$-th character edited (added or deleted). Given $\SR(P, T)$, we can compute $\SR(P', T)$ in $\OO(\log |T|)$ time.
    \label{th:edit}
\end{theorem}

\begin{proof}
    Assume we want to perform a character addition. From \autoref{lem:remove_from_pattern}, we can compute $\SR(P[0,i), T)$ and $\SR(P[i,|P|), T)$ in $\OO(\log |T|)$ time. That is, we split the pattern at the index we want to change. To add some character $c \in \Sigma$, we can then compute $\SR(c, T)$ in $\OO(\log |T|)$ time using binary search, and then concatenate the suffix ranges as outlined in \autoref{th:concat}. In the case of a deletion, we just need to concatenate $\SR(P[0,i), T)$ and $\SR(P[i+1,|P|), T)$.
\end{proof}

So far, we assumed the $P$ pattern always occurs in the text $T$. However, if an edit causes $P$ to no longer appear in $T$, the suffix range becomes empty ($\SR(P, T) = [0,0)$), and we lose the matching information. To handle this, we adopt the idea of an occurrence partition (or cover) from Amir el al.~\cite{amir2007dynamic}.

%% file: text/partition.tex
\section{Dealing with patterns that do not occur}

Now we look at the case when the pattern might not occur in the text. What we can do is represent the pattern as a \textit{concatenation} of strings that occur in the text, which we call an \textit{occurrence partition}, see \autoref{def:partition}.

\nomenclature[11]{$\vert P\vert_T$}{Minimum size of an occurrence partition of $P$ with respect to $T$ (see \autoref{def:partition})}
\begin{definition}[Occurrence Partition.]
    \label{def:partition}
    Let $P$ be a pattern and $T$ be a text. An \textit{occurrence partition} of $P$ with respect to $T$ is a sequence of strings $(P_1, P_2, \allowbreak \dots, \allowbreak P_k)$ that partition $P$, that is, $P = P_1 \circ P_2 \circ \dots \circ P_k$, and satisfies two properties:
    
    \begin{enumerate}
        \item Occurrence: $P_i$ either occurs as a substring of $T$ or $|P_i| = 1$ and $P_i$ represents a character that does not occur in $T$.
        \item Maximality: $P_i \circ P_{i+1}$ does not occur in $T$.
    \end{enumerate}

    We also denote the minimum size of such partition as $|P|_T$.
\end{definition}

Note that not all occurrences partitions are minimum, for example, for $T =$ {\tt aabcaba} and $P = $ {\tt abaaba}, ({\tt ab}, {\tt aab}, {\tt a}) is an occurrence partition, but there is another of size two: ({\tt aba}, {\tt aba}).

\begin{lemma}
    Assuming all characters of $P$ occur in $T$, a pattern $P$ occurs in a text $T$ if, and only if, the occurrence partition of $P$ with respect to $T$ has size one.
    \label{lem:partition_oc}
\end{lemma}

\begin{proof}
    Follows directly from the definition.
\end{proof}

The idea is then to represent an occurrence partition of the pattern, and, for every string in the partition, maintain its suffix range. When asked to return the number of occurrences of the pattern, we return zero if the size of the partition is greater than one, or the length of the suffix range otherwise, as per \autoref{lem:partition_oc} (taking care of characters that do not occur in the text).

We can represent the occurrence partition in a balanced binary search tree, since we want to apply modifications, and this will allow us to do that in $\OO(\log |P|)$ time.

We are left with the task of maintaining the occurrence partition properties when modifying some character. Turns out that this is easy: given some index to edit, we can find the string from the partition that should be edited (by traversing the binary search tree). After that, we can apply \autoref{th:edit}, but not merge the suffix ranges when they would result in an empty range (occurrence property). This might break the maximality property, so we might need to apply some concatenations. It turns out that at most 4 concatenations are needed to ensure the maximality property, as per \autoref{lem:fix_maximality}.

\begin{lemma}
    Let $P$ be a pattern and $(P_1, P_2, \dots, P_k)$ some occurrence partition of $P$ with respect to the text $T$. After a character edition in $P$, we can find an updated occurrence partition using at most 4 concatenations.
    \label{lem:fix_maximality}
\end{lemma}

\begin{proof}
    Assume our edit was in the $i$-th string of the partition, so we had the partition $(\dots,\allowbreak P_{i-2},\allowbreak P_{i-1},\allowbreak P_{i},\allowbreak P_{i+1},\allowbreak P_{i+2},\allowbreak \dots)$. Following our strategy, $P_i$ might be split into 3 parts (in the case of an addition, the middle part represents the new character), say $Q_1, Q_2, Q_3$. It might be the case that $P_{i-1} \circ Q_1 \circ Q_2 \circ Q_3 \circ P_{i+1}$ occurs in the text (4 concatenations are necessary in this case). But this is the worst-case, since, from the maximality property of the initial partition, we know that $P_{i-2} \circ P_{i-1}$ does not occur in the text, so it can not be the case that $P_{i-2} \circ P_{i-1} \circ S$ occurs, for any string $S$. The same applies with $P_{i+1}$ and $P_{i+2}$.
\end{proof}

Let us go through an example. Assume that we have the text $T =$ {\tt cababaa} and the pattern $P = $ {\tt abcaabb}. An occurrence partition of $P$ with respect to $T$ is ({\tt ab}, {\tt c}, {\tt aa}, {\tt b}, {\tt b}). Now assume we want to insert a character {\tt b} in $P$ at index $4$, that is, between the two {\tt a}'s. This will turn $P$ into {\tt abcababb}. To do this, we first need to split the string {\tt aa}, so our representation becomes ({\tt ab}, {\tt c}, {\tt a}, {\tt a}, {\tt b}, {\tt b}); then we find the suffix range of the new character {\tt b}, and insert it to our representation: ({\tt ab}, {\tt c}, {\tt a}, {\tt b}, {\tt a}, {\tt b}, {\tt b}). Now we are left with concatenating adjacent strings in the representation (starting from the {\tt b} we inserted, in both directions), to fix the maximality property. After doing that, the final representation will become ({\tt ab}, {\tt cabab}, {\tt b}), so we performed 4 concatenations. We can see that, since $\texttt{ab} \circ \texttt{c}$ (concatenation of first and second strings in the initial representation) does not occur in $T$, then $\texttt{ab} \circ \texttt{c} \circ \texttt{abab}$ also does not occur in $T$.

\begin{theorem}
    Let $P$ be a pattern and let $T$ be the text. Let $P'$ be $P$ with the $i$-th character edited (added or deleted). If we have the occurrence partition of $P$, we can find an occurrence partition of $P'$ in $\OO(\log |T|)$ time.
    \label{th:edit_full}
\end{theorem}

\begin{proof}
    From \autoref{lem:fix_maximality}, we only need to use \autoref{th:concat} $\OO(1)$ times. Since we can do all necessary operations in the binary search tree used to represent the occurrence partition in $\OO(\log |P|) \subseteq \OO(\log |T|)$ time, we can edit the pattern in $\OO(\log |T|)$ time.
\end{proof}

%% file: text/improving.tex
\section{Improving pattern search}

We now discuss the pattern search operation, that is, set the current pattern as some pattern $P$ given in the input. We assume the current pattern before the operation is empty. Of course, this can be viewed as repeated pattern edition, and therefore can be done in $\OO(|P| \log |T|)$ time. But we can also compute it in $\OO(|P| + |P|_T \log |T|)$ time.

\begin{lemma}
    Greedily taking each time the largest prefix of the pattern that occurs in the text produces an occurrence partition of size $|P|_T$.
    \label{lem:greedy}
\end{lemma}

\begin{proof}
    Any minimum occurrence partition can be transformed to the partition that the greedy algorithm produces, by repeatedly shifting the positions of the divisions between the strings to the right (from left to right).
\end{proof}

\begin{lemma}
    The algorithm from \cite{myers} to compute $\SR(P, T)$ can be implemented in $\OO(\log |T| + \ell)$ time, if $\ell$ is the size of the largest prefix of $P$ that occurs in $T$.
    \label{lem:sr_faster}
\end{lemma}

\begin{proof}
    The only operation inside the binary search that has nontrivial cost is the $lcp$ computation between the pattern and some suffix of the text. But note that the total time that takes is bounded by $\ell$, since every time it runs another iteration it is increasing the size of a prefix of $P$ that occurs in $T$.
\end{proof}

From \autoref{lem:sr_faster}, we can repeatedly run the algorithm from \cite{myers}, and each time from its output we can figure out the size of the string to use for our partition. This is running the greedy algorithm, and from \autoref{lem:greedy} we know that this produces a partition of size $|P|_T$, and adding the total cost we get $\OO(|P|)$ plus $\OO(\log |T|)$ times the size of the partition. Therefore, we can find the occurrence partition in $\OO(|P| + |P|_T \log |T|)$ time.

It turns out we can, given $\SR(P, T)$, compute $\SR(P \circ \text{\tt c}, T)$ for some character $\text{\tt c} \in \Sigma$ in $\OO(\log |\Sigma|)$ time after $\OO(|T|)$ preprocessing, by essentially using the suffix tree (storing, for each of the $\OO(|T|)$ possible ranges, its child ranges, and doing binary search over those -- there can be at most $|\Sigma|$ children).

%% file: text/substring.tex
\section{Substring editions and running time bounds}

Substring modifications can be achieved by operations in the binary search tree that represents the occurrence partition. Deletion and transposition of a substring of the pattern can be done by split and join operations in the binary search tree. By similar arguments as before, we see that only a constant number of merges of suffix ranges are needed after each operation.

Substring copying can be achieved using persistent binary search trees in a similar way, since the persistence of nodes and subtrees allows us to ``copy'' a subtree \cite{driscoll,sarnak}.

The time bounds of our solution for \autoref{prob:dyn_pattern_static_text} are described in \autoref{tab:complexities_static}, along with the time bounds achieved by Amir and Kondratovsky~\cite{amir2018} using suffix trees. Our solution for the static text case is considerably simpler, and we improve on the preprocessing time and space and on the substring editions.

\begin{table}[H]
    \centering
    \caption{Time bounds for \autoref{prob:dyn_pattern_static_text}. $\ell$ is the size of the modified substring.}
    \begin{tabular}{|c|c|c|}
        \hline
        \cellcolor{gray!20} Operation & \cellcolor{gray!20} \cite{amir2018} & \cellcolor{gray!20} This work \\
        \hline\hline
        Preprocess Time and Space & $\OO(|T| \sqrt{\log |T|})$ & $\OO(|T|)$ \\
        \hline
        Pattern Search & $\OO(|P| \log |T|)$ & \makecell{$\OO(|P| + |P|_T \log |T|)$ \\ or \\ $\OO(|P| \log |\Sigma|)$} \\
        % Pattern Search & $\OO(|P| \log |T|)$ & \tiny $\OO\left(|P| + \min\left( |P|_T, \frac{|P|}{\log_{|\Sigma|} |P|} \right) \log |T|\right)$ \\
        \hline
        Pattern Symbol Edition & $\OO(\log |T|)$ & $\OO(\log |T|)$ \\
        \hline
        Pattern Substring Edition & $\OO(\log |T| + \ell)$ & $\OO(\log |T|)$ \\
        \hline
    \end{tabular}
    \label{tab:complexities_static}
\end{table}

%% file: text/online.tex
\section{Online text}

We can extend our solutions for \autoref{prob:dyn_pattern_static_text} to maintain an \textit{online} text, that is, also support text extensions (on one end). We will assume we want to extend the text on the left.

\begin{problem}[Dynamic Pattern and Online Text Matching]
    \label{prob:dyn_pattern_online_text}
    
    \parindent0pt

    \textbf{Input}: A string $T$ that represents the text, several updates to the pattern string $P$ that can be one of the following.

    \begin{enumerate}
        \item Pattern search: computing the representation of a pattern;
        \item Pattern symbol edition: addition or deletion of some character of the pattern;
        \item Pattern substring edition: substring deletion, transposition (moving the substring to another position), or copy.
    \end{enumerate}

    Also, we support adding a new symbol to the left of the text.

    \textbf{Output}: After every update to the pattern, the number of times it occurs in the text.
\end{problem}

To solve \autoref{prob:dyn_pattern_online_text}, we use the online suffix array described in \cite{monteiro}, which maintains the suffix array structure in a balanced binary search tree, allowing for text extensions in amortized logarithmic time. By storing pointers to the relevant tree nodes, we can maintain the suffix ranges. The problem is that, after a symbol extension in the text, the occurrence partition might lose the maximality property.

But note that this is not a problem: we just need to try to concatenate the first and second elements of the occurrence partition after every operation.

This obviously maintains correctness, because we just need to know if the pattern will occur or not, and this will be the case if, and only if, we are able to join all elements of the occurrence partition into one (\autoref{lem:partition_oc}).

Also note that this does not add to the time complexities of the operations, in an amortized sense. The idea is that the number of joins we make is bounded by the number of splits, and we do a constant number of splits per operation. This is formalized in \autoref{th:online}.

\begin{theorem}
    For a text $T$, \autoref{prob:dyn_pattern_online_text} can be solved with $\OO(|T|)$ time construction and $\OO(\log |T|)$ amortized time per operation.
    \label{th:online}
\end{theorem}

\begin{proof}
    We will use a potential function to define our amortized complexities \cite{cormen}. Let $D_i$ be the state of our data structure after operation $i$, and let $E(D_i)$ be the number of elements of the occurrence partition. Define our potential function $\Phi(D_i) = E(D_i) \cdot \log|T|$. %We also denote $\Delta \Phi(D_i) = \Phi(D_i) - \Phi(D_{i-1})$.
    Every operation increases the number of elements by a constant, so this adds an amortized cost of $\OO(\log |T|)$ to each operation. If we assume that we then perform $k$ concatenations, we have $\Phi(D_i) - \Phi(D_{i-1}) = c_1 \log |T| - k \cdot \log |T|$, for some constant $c_1$.
    
    Finally, summing the real cost with the change of our amortized function, we get the amortized cost, for some constant $c_2$:

    \begin{align*}
    & c_2 \log |T| + k \log |T| + \left( \Phi(D_i) - \Phi(D_{i-1}) \right) \\
        =&\, c_2 \log |T| + k \log |T| + c_1 \log |T| - k \cdot \log |T| \\
        =&\, c_2 \log |T| + c_1 \log |T| \\
        \in&\, \OO(\log |T|) \\
    \end{align*}
    
\end{proof}

The time complexities for \autoref{prob:dyn_pattern_online_text} are summarized in \autoref{tab:complexities_online}.

\begin{table}[ht]
    \centering
    \caption{Time bounds for \autoref{prob:dyn_pattern_online_text}. The $\star$ symbol represents amortized bounds, and $\ell$ is the size of the modified substring.}
    \begin{tabular}{|c|c|c|}
        \hline
        \cellcolor{gray!20} Operation & \cellcolor{gray!20} \cite{amir2018} & \cellcolor{gray!20} This work \\
        \hline\hline
        Preprocess Time and Space & $\OO(|T| \sqrt{\log |T|})$ & $\OO(|T|)$ \\
        \hline
        Pattern Symbol Edition & $\OO(\log |T|)$ & $\OO(\log |T|)^\star$ \\
        \hline
        Pattern Substring Edition & $\OO(\log |T| + \ell)$ & $\OO(\log |T|)^\star$ \\
        \hline
        Text Symbol Extension & $\OO(\log |T|)$ & $\OO(\log |T|)^\star$ \\
        \hline
    \end{tabular}
    \label{tab:complexities_online}
\end{table}

%% file: text/conclusion.tex
\section{Conclusion}

In this work, we tackled the \textit{modified pattern reporting problem} defined by Amir and Kondratovsky~\cite{amir2018}. We improved upon their algorithm, with a simpler and faster algorithm -- our preprocess time is linear, and we can support substring editions in sub-linear time (see \autoref{tab:complexities_static}).

In addition, we apply our algorithm for an online text, achieving amortized logarithmic bounds (see \autoref{tab:complexities_online}). The paper by Amir and Kondratovsky~\cite{amir2018} claim to achieve worst-case logarithmic bounds, though their analysis omits key details that prevent independent verification.

To illustrate the feasibility of the proposed solution, it was implemented in C\texttt{++} (available in a public repository\footnote{\url{https://github.com/brunomaletta/DynamicPatternMatching}}) and tested with random test data against a naive solution. Our algorithm performed orders of magnitude faster than the naive algorithm (in the worst case) even for strings with size in the order of one million, suggesting its practicability.

%% file: text/computation.tex
Here we show how to compute $\SR(P, T)$. It turns out that this can be computed in $\OO(|P| + \log |T|)$ \cite{myers}. For this, we compute the first and last indices of the range independently.

Let us first see how to compute the first index of the range in $\OO(|P| \log |T|)$. This can be easily done with a binary search, since what we want is the first suffix that is not smaller than $P$, and the suffixes are sorted. Note that the $|P|$ factor in the complexity comes from string comparison. This is implemented in \autoref{alg:sr_slow}. We are using the type of binary search that maintains that the answer is in the open range $(L, R)$.

\begin{algorithm}[ht]
    \caption{Slow Algorithm for Suffix Range.}
    \hspace*{\algorithmicindent} \textbf{Input}: pattern $P$, suffix array of the text $\SA(T)$.
    \\
    \hspace*{\algorithmicindent} \textbf{Output}: first index of $\SR(P, T)$, or $|T|$ if $P$ does not occur in $T$.
    \\
    \hspace*{\algorithmicindent} \textbf{Time Complexity}: $\OO(|P| \log |T|)$.
    \bigskip
    
    \begin{algorithmic}
        \If{$lcp(P, T_{\SA[0]}) = |P|$}
            \Return 0
        \EndIf
        \\
    
        \State $(L, R) \gets (0, |T|-1)$
        \\

        \While{$R - L > 1$}
            \State $M \gets \left\lfloor \frac{L+R}{2} \right\rfloor$
            \If{$P \preceq T_{\SA[M]}$}
                \State $R \gets M$
            \Else
                \State $L \gets M$
            \EndIf
        \EndWhile
        \\
        
        \If{$lcp(P, T_{\SA[R]}) < |P|$}
            \Return $|T|$ \Comment{$P$ does not occur in $T$}
        \EndIf
        \\
        
        \Return $R$
    \end{algorithmic}
    \label{alg:sr_slow}
\end{algorithm}

For the faster $\OO(|P| + \log |T|)$ implementation, we need to use the $lcp$ information. At every iteration of the binary search, we maintain two values $\ell$ and $r$, the $lcp$ between $P$ and the suffixes $T_{\SA[L]}$ and $T_{\SA[R]}$, respectively.

Now, we need to compute $lcp(P, T_{\SA[M]})$, and then the next character will tell us what side of the binary search we need to go to. For that, let us assume that $\ell \geq r$ (the other case will be symmetric). We now compare $\ell$ with $lcp(T_{\SA[L]}, T_{\SA[M]})$ (which we can compute in constant time from \autoref{lem:lcp_query}). We have two cases:

\begin{enumerate}
    \item $lcp(T_{\SA[L]}, T_{\SA[M]}) < \ell$: in this case, $lcp(P, T_{\SA[M]}) = lcp(T_{\SA[L]}, T_{\SA[M]})$, because $P$ matches $\ell$ characters from $T_{\SA[L]}$, and $T_{\SA[L]}$ matches less than $\ell$ characters from $T_{\SA[M]}$.
    \label{case:sa_range_1}

    \item $lcp(T_{\SA[L]}, T_{\SA[M]}) \geq \ell$: now we can find the smallest $i > \ell$ such that $P[i] \neq T_{\SA[M]}[i]$, and assign $lcp(P, T_{\SA[M]}) = i$.
    \label{case:sa_range_2}
\end{enumerate}

It turns out that if we do \autoref{case:sa_range_2} naively and update our $(\ell, r)$ accordingly, the total cost of these comparisons throughout the binary search is $\OO(|P|)$, because each time we increase $\max(\ell, r)$, and these values can only increase up to $|P|$. Also, on case \autoref{case:sa_range_1} we always go left on the binary search (because $r \leq lcp(T_{\SA[L]}, T_{\SA[M]}) < \ell$), thus maintaining or increasing the value of $r$. In fact, both $\ell$ and $r$ never decrease.

Therefore the total time cost of the binary search is $\OO(|P| + \log |T|)$ (see \autoref{alg:sr_fast}).

\begin{algorithm}[ht]
    \caption{Fast Algorithm for Suffix Range.}
    \hspace*{\algorithmicindent} \textbf{Input}: pattern $P$, suffix array of the text $\SA(T)$.
    \\
    \hspace*{\algorithmicindent} \textbf{Output}: first index of $\SR(P, T)$, or $|T|$ if $P$ does not occur in $T$.
    \\
    \hspace*{\algorithmicindent} \textbf{Time Complexity}: $\OO(|P| + \log |T|)$.
    \bigskip
    
    \begin{algorithmic}
        \State $\ell \gets lcp(P, T_{\SA[0]})$
        \State $r \gets lcp(P, T_{\SA[|T|-1]})$
        \If{$\ell = |P|$}
            \Return 0
        \EndIf
        \\
        
        \State $(L, R) \gets (0, |T|-1)$
        \\
        
        \While{$R - L > 1$}
            \State $M \gets \left\lfloor \frac{L+R}{2} \right\rfloor$
            \If{$\ell \geq r$}
                \If{$lcp(T_{\SA[L]}, T_{\SA[M]}) < \ell$}
                    \State $m \gets lcp(T_{\SA[L]}, T_{\SA[M]})$ \Comment Computed in constant time
                \Else
                    \State $m \gets \ell + lcp(P_\ell, T_{\SA[M]+\ell})$ \Comment Computed naively
                \EndIf
            \Else
                \If{$lcp(T_{\SA[R]}, T_{\SA[M]}) < r$}
                    \State $m \gets lcp(T_{\SA[R]}, T_{\SA[M]})$ \Comment Computed in constant time
                \Else
                    \State $m \gets r + lcp(P_r, T_{\SA[M]+r})$ \Comment Computed naively
                \EndIf
            \EndIf
            \\

            \If{$m = |T|$ or ($\SA[M]+m < |T|$ and $P[m] < T_{\SA[M]}[m]$)}
                \State $(R, r) \gets (M, m)$
            \Else
                \State $(L, \ell) \gets (M, m)$
            \EndIf
        \EndWhile
        \\

        \State $plcp \gets \max(\ell, r)$ \Comment Largest prefix of $P$ that occurs in $T$
        \If{$plcp < |P|$}
            \Return $|T|$
        \EndIf
        \\
        
        \Return $R$
    \end{algorithmic}
    \label{alg:sr_fast}
\end{algorithm}

%% file: samplepaper.bbl
\begin{thebibliography}{10}
\providecommand{\url}[1]{\texttt{#1}}
\providecommand{\urlprefix}{URL }
\providecommand{\doi}[1]{https://doi.org/#1}

\bibitem{ahocorasick}
Aho, A.V., Corasick, M.J.: Efficient string matching: an aid to bibliographic search. Communications of the ACM  \textbf{18}(6),  333--340 (1975)

\bibitem{alstrup2000pattern}
Alstrup, S., Brodal, G.S., Rauhe, T.: Pattern matching in dynamic texts. In: Proceedings of the eleventh annual ACM-SIAM symposium on Discrete algorithms. pp. 819--828 (2000)

\bibitem{amir2018}
Amir, A., Kondratovsky, E.: Searching for a modified pattern in a changing text. In: International Symposium on String Processing and Information Retrieval. pp. 241--253. Springer (2018)

\bibitem{amir2005}
Amir, A., Kopelowitz, T., Lewenstein, M., Lewenstein, N.: Towards real-time suffix tree construction. In: International Symposium on String Processing and Information Retrieval. pp. 67--78. Springer (2005)

\bibitem{amir2007dynamic}
Amir, A., Landau, G.M., Lewenstein, M., Sokol, D.: Dynamic text and static pattern matching. ACM Transactions on Algorithms (TALG)  \textbf{3}(2),  19--es (2007)

\bibitem{bender}
Bender, M.A., Farach-Colton, M.: The {LCA} problem revisited. In: LATIN 2000: Theoretical Informatics: 4th Latin American Symposium, Punta del Este, Uruguay, April 10-14, 2000 Proceedings 4. pp. 88--94. Springer (2000)

\bibitem{cole2006}
Cole, R., Kopelowitz, T., Lewenstein, M.: Suffix trays and suffix trists: structures for faster text indexing. In: International Colloquium on Automata, Languages, and Programming. pp. 358--369. Springer (2006)

\bibitem{cole2003}
Cole, R., Lewenstein, M.: Multidimensional matching and fast search in suffix trees. In: Proceedings of the fourteenth annual ACM-SIAM symposium on Discrete algorithms. pp. 851--852 (2003)

\bibitem{cormen}
Cormen, T.H., Leiserson, C.E., Rivest, R.L., Stein, C.: Introduction to algorithms. MIT press (2022)

\bibitem{das2022internal}
Das, R., He, M., Kondratovsky, E., Munro, J.I., Wu, K.: Internal masked prefix sums and its connection to fully internal measurement queries. In: International Symposium on String Processing and Information Retrieval. pp. 217--232. Springer (2022)

\bibitem{driscoll}
Driscoll, J.R., Sarnak, N., Sleator, D.D., Tarjan, R.E.: Making data structures persistent. Journal of computer and system sciences  \textbf{38}(1),  86--124 (1989)

\bibitem{farach}
Farach, M.: Optimal suffix tree construction with large alphabets. In: Proceedings 38th Annual Symposium on Foundations of Computer Science. pp. 137--143. IEEE (1997)

\bibitem{ferragina1998}
Ferragina, P., Grossi, R.: Optimal on-line search and sublinear time update in string matching. SIAM Journal on Computing  \textbf{27}(3),  713--736 (1998)

\bibitem{fischer}
Fischer, M.J., Paterson, M.S.: String matching and other products. In: Complexity of Computation, RM Karp (editor), SIAM-AMS Proceedings. vol.~7, pp. 113--125 (1974)

\bibitem{gu}
Gu, M., Farach, M., Beigel, R.: An efficient algorithm for dynamic text indexing. In: Proceedings of the Fifth Annual ACM-SIAM Symposium on Discrete Algorithms. p. 697–704. SODA '94, Society for Industrial and Applied Mathematics, USA (1994)

\bibitem{ivanov}
Ivanov, A.: Distinguishing an approximate word’s inclusion on turing machine in real time. Izv. Acad. Nauk USSR Ser. Mat  \textbf{48},  520--568 (1984)

\bibitem{karkkainen}
K{\"a}rkk{\"a}inen, J., Sanders, P.: Simple linear work suffix array construction. In: International colloquium on automata, languages, and programming. pp. 943--955. Springer (2003)

\bibitem{kasai}
Kasai, T., Lee, G., Arimura, H., Arikawa, S., Park, K.: Linear-time longest-common-prefix computation in suffix arrays and its applications. In: CPM. vol.~2089, pp. 181--192. Springer (2001)

\bibitem{kim}
Kim, D.K., Sim, J.S., Park, H., Park, K.: Constructing suffix arrays in linear time. Journal of Discrete Algorithms  \textbf{3}(2-4),  126--142 (2005)

\bibitem{kmp}
Knuth, D.E., Morris, Jr, J.H., Pratt, V.R.: Fast pattern matching in strings. SIAM journal on computing  \textbf{6}(2),  323--350 (1977)

\bibitem{kociumaka2024internal}
Kociumaka, T., Radoszewski, J., Rytter, W., Wale{\'n}, T.: Internal pattern matching queries in a text and applications. SIAM Journal on Computing  \textbf{53}(5),  1524--1577 (2024)

\bibitem{landau}
Landau, G.M., Vishkin, U.: Efficient string matching with $k$ mismatches. Theoretical Computer Science  \textbf{43},  239--249 (1986)

\bibitem{myers}
Manber, U., Myers, G.: Suffix arrays: a new method for on-line string searches. {SIAM} Journal on Computing  \textbf{22}(5),  935--948 (1993)

\bibitem{monteiro}
Monteiro, B.: String Matching with a Dynamic Pattern. Master's thesis, Universidade Federal de Minas Gerais (2024)

\bibitem{mp}
Morris~Jr, J., Pratt, V.: A linear pattern-matching algorithm. University of California, Berkeley (1970)

\bibitem{sarnak}
Sarnak, N., Tarjan, R.E.: Planar point location using persistent search trees. Communications of the ACM  \textbf{29}(7),  669--679 (1986)

\bibitem{weiner}
Weiner, P.: Linear pattern matching algorithms. In: 14th Annual Symposium on Switching and Automata Theory (swat 1973). pp. 1--11. IEEE (1973)

\end{thebibliography}
